\documentclass[12pt]{article}

\def\hybrid{\topmargin -20pt  \oddsidemargin 0pt
      \headheight 0pt   \headsep 0pt
      \textwidth 6.25in 
      \textheight 9.5in 
      \marginparwidth .875in
      \parskip 5pt plus 1pt   \jot = 1.5ex}

\hybrid

\usepackage[latin1]{inputenc}
\usepackage[T1]{fontenc}
\usepackage{amsmath}
\usepackage{amsthm}
\usepackage{amssymb}
\usepackage[all]{xy}
\usepackage{amscd,amsmath,amssymb,euscript}

\numberwithin{equation}{section}

\newtheorem{theorem}{Theorem}[section]
\newtheorem{lemma}[theorem]{Lemma}

\newtheorem{corollary}[theorem]{Corollary}

\newcommand{\bbC}{\mathbb{C}}
\newcommand{\bbZ}{\mathbb{Z}}
\newcommand{\bbP}{\mathbb{P}}
\newcommand{\rmH}{\mathrm{H}}
\newcommand{\calO}{\mathcal{O}}
\newcommand{\calE}{\mathcal{E}}
\newcommand{\longto}{\longrightarrow}

\newcommand{\Fix}{\mathrm{Fix}}
\newcommand{\id}{\mathrm{id}}
\newcommand{\rk}{\mathrm{rk}}
\def\dbar{\bar\partial}

\newcommand{\D}{\bar{D}}
\DeclareMathOperator{\Ext}{Ext}
\newcommand{\beqa}{\begin{equation}}
\newcommand{\eeqa}{\end{equation}}
\theoremstyle{definition}

\renewcommand{\labelenumi}{\theenumi}
\renewcommand{\theenumi}{(\roman{enumi})}

\begin{document}
\thispagestyle{empty}

\rightline{}
\vspace{2truecm}
\centerline{\bf  \Large $SU(5)$ Heterotic Standard Model Bundles}

\vspace{0.5truecm}

\renewcommand{\thefootnote}{}

\centerline{Bj\"orn Andreas and Norbert Hoffmann}

\footnotetext{This work was supported by the SFB/647 ``Space-Time-Matter. Analytic and Geometric Structures'' of the DFG (German Research Foundation). B. A. is supported by project DFG-SFB 647/A3, N. H. is supported by project DFG-SFB 647/A11.}

\vspace{.6truecm}

\centerline{{\em Institut f\"ur Mathematik, Freie Universit\"at Berlin}}
\centerline{{\em Arnimallee 3, 14195 Berlin, Germany}}

\begin{abstract}\noindent
We construct a class of stable $SU(5)$ bundles on an elliptically fibered Calabi-Yau
threefold with two sections, a variant of the ordinary Weierstrass fibration, which admits
a free involution. The bundles are invariant under the involution, solve the topological 
constraint imposed by the heterotic anomaly equation and give three generations
of Standard Model fermions after symmetry breaking by Wilson lines of the intermediate $SU(5)$ GUT-group to the Standard Model gauge group. Among the solutions we find some which can be perturbed 
to solutions of the Strominger system. Thus these solutions provide a step toward the construction of phenomenologically realistic heterotic flux compactifications via non-K\"ahler deformations of Calabi-Yau geometries with bundles. This particular class of solutions involves a rank two hidden sector bundle and does not require background fivebranes for anomaly cancellation. 
\end{abstract}

\newpage

\section{Introduction}
Compactifications of the $E_8\times E_8$ heterotic string on non-simply connected Calabi-Yau threefolds
with ${\mathbb Z}_2$ fundamental group provide one possibility to obtain the Standard Model of elementary particle physics in the low energy limit of string theory. These compactifications require 
the specification of a stable $SU(5)$ gauge bundle $V$ which breaks the visible $E_8$ gauge group
to an $SU(5)$ Grand Unified Theory (GUT) group. A ${\mathbb Z}_2$ Wilson line is then used to 
break the intermediate GUT group to the Standard Model gauge group $SU(3)\times SU(2)\times U(1)$. 
In addition to $V$ one has to specify a second stable $SU(n)$ bundle $V_{hid}$ which is used to break the hidden sector $E_8$ gauge group. However, $V_{hid}$ is often taken to be trivial and so 
an unbroken hidden sector $E_8$ gauge group is left.

One way to realize such a compactification is to start first with a simply connected Calabi-Yau threefold $X$ which admits a free involution $\tau_X$ and construct a polystable bundle $V\oplus V_{hid}$ on $X$
which is invariant under $\tau_X$ and with second Chern classes satisfying the topological 
constraint, eqn.\ref{topano} below, imposed by the heterotic anomaly equation and third Chern classes
satisfying the phenomenological constraint eqn.\ref{matter}. The quotient $X/{\tau_X}$ has then the required ${\mathbb Z}_2$ fundamental group and as the bundle $V\oplus V_{hid}$ is $\tau_X$-invariant
it descends to a bundle on $X/{\tau_X}$ giving three net-generations of Standard Model fermions. In summary, these compactifications state the mathematical problem to specify
\begin{itemize}
\item a smooth Calabi-Yau threefold $X$ which admits a free involution $\tau_X\colon X\to X$,
\item a stable $\tau_X$-invariant $SU(5)$ vector bundle $V$ and a stable $\tau_X$-invariant $SU(n)$ vector bundle $V_{hid}$ (with $n\leq 8$) on $X$,
\end{itemize}
which have to satisfy
\begin{enumerate}
\item $c_2(X)-c_2(V)-c_2(V_{hid})=[W]$ is an effective curve class (or zero),\label{topano}
\item $c_3(V)/2=\pm 6$, \ \ \ $c_3(V_{hid})=0.$ \label{matter}
\end{enumerate}
In the past decade there has been an intensive search for triples $(X,V,V_{hid})$ which satisfy these constraints. A solution to these constraints is given by the Schoen Calabi-Yau threefold $X$ and $V$ a stable extension bundle of ${\mathbb Z}_2$-invariant spectral cover bundles \cite{Don1}, \cite{Don2}. A version of this model leads precisely to the Minimal Supersymmetric Standard Model (MSSM) with no exotic matter \cite{BDon3} and a single pair of Higgs. 

Another class of Calabi-Yau threefolds which admits a free involution has been constructed in \cite{ACK99}. By contrast to the elliptically fibered Calabi-Yau threefolds with one section, as studied for instance in \cite{FMW}, the Calabi-Yau threefolds in \cite{ACK99} use a specific elliptic fibration type
which has two sections and allows to construct a free involution. There have been various attempts
\cite{ACK99}, \cite{AC1}, \cite{AC2} to construct invariant stable bundles (using various bundle construction methods) on this class of Calabi-Yau threefolds but none of them succeeded completely and solved the above constraints. For instance, in \cite{AC1} either the Standard Model gauge group 
times an additional $U(1)$, or just the Standard Model gauge group but with additional exotic matter
has been obtained; in \cite{AC2} a rank five bundle has been constructed via deformation of the direct
sum of a stable rank four bundle and the trivial bundle, however, only some necessary conditions for the existence of such a stable invariant deformation have been checked. 

The primary goal of this note is to construct a class of stable $\tau_X$-invariant $SU(5)$ bundles on the class of manifolds of \cite{ACK99} which satisfy all of the above constraints. To construct the bundles we will apply the bundle construction method of \cite{AH}. The bundles are given by non-trivial extensions of $\tau_X$-invariant $SU(4)$ bundles by $\tau_X$-invariant line bundles.

The question if a given solution to (i) can be modified such that it actually solves the anomaly equation on the level of differential forms has been another motivation for this note. (For this recall that (i) is just the integrability condition for the existence of a solution to equation (\ref{anom}) below.) This question is naturally embedded into the search for heterotic flux compactifications which are characterized by a Hermitian structure with $(1,1)$-form $\omega$, a holomorphic $(3,0)$-form $\Omega$ and a gauge bundle on the internal manifold.
Supersymmetry and anomaly cancellation impose the conditions: the internal manifold has to be conformal balanced (\ref{balanced}), the gauge bundle has to satisfy the Hermitian-Yang-Mills equation (\ref{DUY}) and the flux has to satisfy the anomaly equation (\ref{anom}) (cf.~\cite{Strom},\cite{FuYau})
\begin{align}
d(||\Omega||_\omega \omega^2) & = 0, \label{balanced}\\
F^{2,0}=F^{0,2}=0, \ \ \  F\wedge \omega^2&=0,\label{DUY}\\
i\partial\dbar\omega - \alpha'(tr(R\wedge R)-tr(F\wedge F))&=0\label{anom},
\end{align}
where $R$ is the curvature of a unitary connection on the tangent bundle of the internal manifold. This system of equations is usually called {\it the Strominger system} (note in (\ref{anom}) the case $[W]=0$ 
is assumed as otherwise (\ref{anom}) receives a contribution of a current which integrates
to one in the direction transverse to a single curve wrapped by a fivebrane). 

First examples of solutions of the Strominger system have been obtained in \cite{Yau1, FuYau, FTY, LiYau}. One difficulty in obtaining smooth solutions lies in the fact that in general $\omega$ is not closed and many theorems of K\"ahler geometry and thus methods of algebraic geometry do not apply. One approach to obtain solutions is to simultaneously perturb a Calabi-Yau threefold and a polystable bundle over it and so avoid the direct construction of non-K\"ahler manifolds with stable bundles. This approach has been used in \cite{LiYau} where it is shown that a deformation of the holomorphic structure on the direct sum of the tangent bundle and the trivial bundle of a given Calabi-Yau threefold leads to a smooth solution of the Strominger system whereas the original Calabi-Yau space is perturbed to a non-K\"ahler space. Inspired by the method developed in \cite{LiYau}, this result has been extended in \cite{AGF1}, \cite{AGF2} to the case of an arbitrary polystable bundle $W$ over a Calabi-Yau threefold which satisfies 
$c_2(X)=c_2(W)$. A result along these lines has been originally conjectured using a different framework in \cite{Wi86} with evidence given in \cite{WiWi87,WuWi87}.

Thus if $c_2(X)=c_2(V)+c_2(V_{hid})$, in condition (i) above, we can apply the results of \cite{AGF1}, \cite{AGF2} and obtain solutions of the Strominger system. 

In section 2 we briefly review the relevant geometrical properties of the cover Calabi-Yau threefold $X$, the action of the free involution $\tau_X$ and the geometry of quotient space $X/\tau_X$ which has the required ${\mathbb Z}_2$ fundamental group. For more details and proofs we refer to \cite{ACK99} and \cite{AC2}. In section 3, we give a broad outline of the bundle construction of \cite{AH} which we will apply to construct a class of stable $\tau_X$-invariant $SU(5)$ bundles. In section 4 we begin with constructing stable $\tau_B$-invariant $SU(2)$ bundles on the Hirzebruch surface $B = \bbP^1 \times \bbP^1$. These bundles serve as
input bundles for the bundle extension in the next section. In section 5 we construct stable $\tau_X$-invariant $SU(4)$ bundles. In section 6 we construct a class of stable $\tau_X$-invariant $SU(5)$ bundles. In section 7 
we show that these bundles are capable to solve the constraints (i) and (ii). In section 8 we will specify a $\tau_X$-invariant hidden sector $SU(2)$ bundle and solve the anomaly constraint without invoking a number of background fivebranes. This solution can be deformed to 
a solution of the Strominger system. 

\section{The Calabi-Yau threefold and its quotient }

In this section we briefly review the main geometrical properties of the cover Calabi-Yau threefold and of its quotient. For more details we refer to \cite{ACK99} and \cite{AC1}.

Restricting to elliptic Calabi-Yau threefolds $X$, we search for a free involution $\tau_X$ which preserves the fibration structure and holomorphic threeform of $X$.
If there is some involution preserving the fibration structure then this must project to some (not necessarily free acting) involution $\tau_B$ in the base $B$.
In order to realize this it turns out that we have to require that the elliptically fibered Calabi-Yau threefold admits two sections $\sigma_1$ and $\sigma_2=\tau_X^*\sigma_1$.
Two possibilities to realize an elliptically fibered Calabi-Yau threefold with two sections have been investigated.
One possibility is to search for an elliptically fibered Calabi-Yau threefold with a changed type of elliptic fiber so that the global fibration has then besides the usually assumed single section a second one \cite{ACK99}.
Alternatively, one can specialize the Weierstrass model to force a second section and resolve a curve of $A_1$ singularities that occur in this process \cite{Don99}.
\vskip 0.2cm
\noindent{\bf The cover Calabi-Yau threefold $X$ with two sections}

We consider a Calabi-Yau threefold $\pi:X\to B$ elliptically fibered over the Hirzebruch surface $B={\mathbb P}^1\times {\mathbb P}^1$.
Let $X$ be the closed subvariety in the weighted projective space bundle $\bbP_{1,2,1}( K_B^{-1} \oplus K_B^{-2} \oplus \calO_B)$ over $B$ given by a generalized Weierstrass equation
\beqa\label{weier}
  y^2 = x^4+ax^2z^2+bxz^3+cz^4
\eeqa
of weighted degree $4$, where $[x:y:z]$ are weighted homogenous coordinates on $\bbP_{1, 2, 1}$ and the coefficients $a,b,c$ are appropriate sections of $K_B^{-n}$ with $n=2,3,4$, respectively.

(More precisely, we let $\bbC^*$ act linearly with weights $1, 2, 1$ on the vector bundle
\beqa
  \calE := K_B^{-1} \oplus K_B^{-2} \oplus \calO_B \longto B
\eeqa
and denote by $x, y, z$ the projections of $\calE$ onto its summands $K_B^{-1}, K_B^{-2}, \calO_B$, respectively.
Then the cone $\hat{X}$ over $X$ is the inverse image of the zero section under the morphism
\beqa
  y^2 - x^4 - ax^2z^2 - bxz^3 - cz^4: \calE \longto K_B^{-4}
\eeqa
of varieties over $B$. The subvariety $\hat{X} \subseteq \calE$ is $\bbC^*$-invariant, and $X$ is by definition the quotient of $\hat{X}$ minus the zero section modulo $\bbC^*$.)

Asuming that $a, b, c$ are sufficiently generic, $X$ is a Calabi-Yau threefold \cite{ACK99}, and the fibration $\pi: X \to B$ admits two cohomologically inequivalent sections $\sigma_1, \sigma_2: B \to X$;
in each fiber, they are given by the two points $[x:y:z] = [1 : \pm 1 : 0]$ in $\bbP_{1,2,1}$.

Let $c_i := c_i( B)$ denote the Chern classes of $B$. We denote the class of the divisor $\sigma_{\nu}( B)$ in $X$ again by $\sigma_{\nu}$, and put $\Sigma:=\sigma_1+\sigma_2$.
Then $\sigma_1 \cdot \sigma_2 = 0$; we also note the adjunction relations $\sigma_{\nu}^2=-\pi^*c_1\cdot \sigma_{\nu}$ and $\Sigma^2=-\pi^*c_1\cdot \Sigma$.
The Chern classes of $X$ are given by \cite{ACK99} 
\beqa
c_1(X)=0, \ \ \ c_2(X)=6\pi^*c_1\cdot \Sigma+\pi^*(c_2+5c_1^2), \ \ \  c_3(X)=-36\pi^*c_1^2.
\eeqa
The Hodge numbers and Euler characteristic of $X$ are given by
\beqa
h^{1,1}(X)=4, \ \ \ h^{2,1}(X)=148, \ \ \ e(X)=-288.
\eeqa
Let us now recall the procedure of \cite{ACK99} used to obtain a free involution $\tau_X$ on $X$.
\vskip 0.2cm
\noindent{\bf The involution $\tau_X$}

We start with the involution $\tau: z \mapsto -z$ of $\bbP^1 = \bbC \cup \{ \infty\}$, and let $\tau_B = \tau \times \tau$ denote the induced involution of $B = \bbP^1 \times \bbP^1$.
In local affine coordinates, $\tau_B$ is thus given by 
\beqa
\tau_B \colon (z_1,z_2) \mapsto (-z_1,-z_2),
\eeqa
and the fixed point set $\Fix(\tau_B)$ consists of the $4$ points $(0,0), (0,\infty), (\infty, 0)$ and $(\infty,\infty)$. 

$\tau_B$ induces an involution on the line bundle $K_B^{-n}$ over $B$, and hence also on its global sections; we still denote these induced involutions by $\tau_B$ and note that the diagram
\beqa \xymatrix{
  K_B^{-n} \ar[d] \ar[r]^{\tau_B} & K_B^{-n} \ar[d]\\
  B \ar[r]^{\tau_B} & B
} \eeqa
commutes.
Explicitly, a basis for the global sections of $K_B^{-n} \cong \calO_{\bbP^1 \times \bbP^1}( 2n, 2n)$ is given by the monomials $z_1^p z_2^q$ over $\bbC \times \bbC \subset B$ with $0 \leq p, q \leq 2n$, and $\tau_B$ acts on them as
\beqa
  \tau_B \colon z_1^p z_2^q \mapsto (-1)^{p+q} z_1^p z_2^q.
\eeqa

If we want to lift $\tau_B$ to a free involution $\tau_X$ on $X$, then we need in particular a free involution on the fibers over $\Fix( \tau_B)$.
A candidate for this is given by the involution
\beqa
  [x:y:z] \mapsto [-x:-y:z]
\eeqa
on $\bbP_{1, 2, 1}$. This map globalizes to the $\bbC^*$-equivariant involution
\beqa
  \tau_{\calE} := (-\tau_B) \oplus (-\tau_B) \oplus \tau_B \quad\text{on}\quad \calE = K_B^{-1} \oplus K_B^{-2} \oplus \calO_B,
\eeqa
which again makes the following diagram commute:
\beqa \xymatrix{
  \calE \ar[d] \ar[r]^{\tau_{\calE}} & \calE \ar[d]\\
  B \ar[r]^{\tau_B} & B
} \eeqa
In order to ensure $\tau_{\calE}( \hat{X}) \subseteq \hat{X}$, we require that the sections $a, b, c$ satisfy
\beqa
  \tau_B( a) = a, \quad \tau_B( b) = -b, \quad \tau_B( c) = c,
\eeqa
which means explicitly that $b$ contains only monomials $z_1^pz_2^q$ with $p+q$ odd, whereas $a$ and $c$ contain only such with $p+q$ even.
Then $\tau_{\calE}$ restricts to an involution on $\hat{X}$, which is still $\bbC^*$-equivariant and descends to an involution $\tau_X$ on $X$ such that the diagram 
\beqa \xymatrix{
  X \ar[d]_{\pi} \ar[r]^{\tau_X} & X \ar[d]^{\pi}\\
  B \ar[r]^{\tau_B} & B
} \eeqa
commutes. The discriminant $\Delta$ of the elliptic fibration remains generic as enough terms in $a,b,c$ survive, so $X$ is still smooth, and we can moreover assume that $\Delta$ is disjoint from $\Fix(\tau_B)$. 
Then $\tau_X$ is actually free, since the fibers of $X$ over $\Fix( \tau_B) \subset B$ are smooth elliptic curves and $\tau_X$ acts on them as a translation by a $2$-torsion point.

For such specialized elliptic fibrations $X$ over $B = \bbP^1 \times \bbP^1$, we thus obtain a free involution $\tau_X$ on $X$ over the involution $\tau_B$ on $B$; see \cite{ACK99} for more details.
\vskip 0.2cm
\noindent{\bf The quotient Calabi-Yau space $X/{\tau_X}$}

The holomorphic threeform on $X$ is $\tau_X$-invariant, so $X/\tau_X$ is a smooth Calabi-Yau threefold with $\pi_1(X/\tau_X)={\mathbb Z}_2$,
and its Hodge numbers and Euler characteristic are
\beqa
h^{1,1}(X/\tau_X)=3, \ \ \ h^{2,1}(X/\tau_X)=75, \ \ \ e(X/\tau_X)=-144,
\eeqa
as is shown in \cite{ACK99}. Note that one K\"ahler class is lost: all divisor classes on $B = \bbP^1 \times \bbP^1$ are $\tau_B$-invariant, but $\tau_X$ exchanges the two sections $\sigma_1$ and $\sigma_2$ of $X$ over $B$,
so only multiples of their sum $\Sigma$ are $\tau_X$-invariant divisor classes on $X$. We also note that the number of complex structure deformations drops due to the special choice of $a,b,c$.

\section{Outline of the bundle construction}
The stable $\tau_X$-invariant $SU(5)$ bundle $V_5$ on $X$ will be constructed as a non-trivial extension. For this we will use the bundle construction of \cite{AH} which we will now briefly summarize.

Let $D$ be a divisor class on $X$. For vector bundle $V$ over $X$, we use the standard notation
\beqa
V(D):=V\otimes{\mathcal O}_X(D).
\eeqa
Now let $H$ be an ample divisor class on $X$, and suppose that two $H$-stable holomorphic vector bundles $V_1$ and $V_2$ over $X$ with $c_1( V_{\nu}) = 0$ are already given.
Put $r_{\nu} := {\rm rank}( V_{\nu})$, $r_{\nu}' := r_{\nu}/\gcd( r_1, r_2)$, $r' := r_1' + r_2'$ and $r := r_1 + r_2$. Then every extension
\begin{equation} \label{extF}
  0 \to V_1( r_2' D) \to V \to V_2( -r_1' D) \to 0
\end{equation}
satisfies ${\rm rank}(V) = r$ and $c_1(V) = 0$. In \cite{AH} it is shown that 
the vector bundle $V$ is $(H + \epsilon D)$-stable for each sufficiently small $\epsilon > 0$ if the following three conditions hold:

\begin{enumerate}
\renewcommand{\labelenumi}{\textbf{(A)}}
\renewcommand{\theenumi}{\textbf{(A)}}
\item \label{A}
    $D\cdot H^2=0$.
\renewcommand{\labelenumi}{\textbf{(B)}}
\renewcommand{\theenumi}{\textbf{(B)}}
\item \label{B}
    The extension (\ref{extF}) does not split.
\renewcommand{\labelenumi}{\textbf{(C)}}
\renewcommand{\theenumi}{\textbf{(C)}}
\item \label{C}
  $D\cdot H\not\equiv 0$ numerically.
\end{enumerate}
Note that this stability result applies to any Calabi-Yau threefold with $h^{1, 1}(X)>1$ \cite{AH}.

Assume \ref{A} holds then \ref{B} can be satisfied if the Euler characteristic $\chi_D( V_2, V_1)<0$
(cf. Lemma 3.3, \cite{AH}) with
\beqa
  \chi_D( V_2, V_1) := \sum_{i = 0}^3 (-1)^i \dim \Ext^i \big( V_2( -r_1' D), V_1( r_2' D) \big).
\eeqa
The Euler characteristic $\chi_D( V_2, V_1)$ is given by the formula
  \begin{equation} \label{eq:HRR}
    \chi_D( V_2, V_1) = \frac{r_1 r_2 r'^3}{6} D^3 + r' \big( \frac{r_1 r_2}{12} c_2( X) - r_2 c_2( V_1) - r_1 c_2( V_2) \big) \cdot D + \frac{r_2}{2} c_3( V_1) - \frac{r_1}{2} c_3( V_2).
  \end{equation}
Now suppose that $V_1$ and $V_2$ are $\tau_X$-invariant, which means that there are isomorphisms $\phi_{\nu}: V_{\nu} \to \tau_X^*( V_{\nu})$ of vector bundles over $X$. The automorphism $\tau_X^*( \phi_{\nu}) \circ \phi_{\nu}$
of $V_{\nu}$ is a scalar since $V_{\nu}$ is stable; multiplying $\phi_{\nu}$ by a square root of this scalar, we can achieve $\tau_X^*( \phi_{\nu}) \circ \phi_{\nu} = \id$ without loss of generality. 
If the divisor class $D$ is also $\tau_X$-invariant, then the same applies to the stable vector bundles $V_1( r_2' D)$ and $V_2( -r_1' D)$; this shows that we can lift $\tau_X$ to involutions on these vector bundles over $X$.
These induce an action of $\tau_X$ on the vector space $\Ext^1 \big( V_2( -r_1' D), V_1( r_2' D) \big)$, and hence an eigenspace decomposition
\beqa
  \Ext^1 \big( V_2( -r_1' D), V_1( r_2' D) \big) = \Ext^1_+ \big( V_2( -r_1' D), V_1( r_2' D) \big) \oplus \Ext^1_- \big( V_2( -r_1' D), V_1( r_2' D) \big).
\eeqa
Our assumption \ref{B} yields $\Ext^1_+ \neq 0$ or $\Ext^1_- \neq 0$. But replacing the lifted involution on one of the vector bundles by its negative exchanges $\Ext^1_+$ and $\Ext^1_-$.
Hence we can achieve $\Ext^1_+ \neq 0$ without loss of generality. Choosing the extension \eqref{extF} in such a way that its class is a nonzero element in $\Ext^1_+$,
it follows that the vector bundle $V$ is also $\tau_X$-invariant. This shows that under our assumptions, nontrivial $\tau_X$-invariant extensions \eqref{extF} always exist. The argument has been also used 
in the construction of \cite{Don1}.

We will now describe in broad outline our procedure to construct the stable $\tau_X$-invariant 
$SU(5)$ bundle. The construction requires three steps. First, we will construct stable $\tau_B$-invariant $SU(2)$ bundles $E_i$ on the base $B={\mathbb P}^1\times {\mathbb P}^1$. The pullback bundles $\pi^*E_i$ (which are shown to be stable) will then be used to construct a $\tau_X$-invariant $SU(4)$ bundle on $X$ applying the above construction. We then proceed and construct the required $SU(5)$ bundle again via the above method of bundle extension. 

Note that the reason for this three step construction is that if one would construct 
the rank five bundle in one step, say as an extension of a rank 2 and a rank 3 bundle, 
then the resulting second Chern classes do not satisfy the anomaly constraint on $X$.
This was already observed in \cite{AC1}. 

\section{Stable $\tau_B$-invariant rank two bundles}
In this section, we construct some stable vector bundles of rank $2$ over $B = \bbP^1 \times \bbP^1$. The method is inspired by the work of Brosius \cite{Bro} on vector bundles over ruled surfaces.
\begin{lemma}
  Let $E$ be a vector bundle of rank $2$ over $B$ with $c_1( E) = 0$ and $\rmH^0( E) = 0$. Suppose that the vector bundles $E|_{\{z_1\} \times \bbP^1}$ and $E|_{\bbP^1 \times \{z_2\}}$ over $\bbP^1$ are trivial
  for general points $z_1, z_2 \in \bbP^1$. Then $E$ is slope $H_B$-stable for every ample divisor class $H_B$ on $B$.
\end{lemma}
\begin{proof}
  Let $L \subset E$ be a coherent subsheaf of rank $1$ such that $E/L$ is torsionfree.
  The induced map of biduals $L^{**} \to E^{**} \cong E$ embeds the torsion sheaf $L^{**}/L$ into the torsionfree sheaf $E/L$, so $L^{**}/L = 0$; this shows that $L$ is locally free, so $L \cong \calO_{\bbP^1 \times \bbP^1}( a, b)$.

  Let $z_1 \in \bbP^1$ be a general point such that $E/L$ is locally free near $\{z_1\} \times \bbP^1$. Then $L|_{\{z_1\} \times \bbP^1}$ is a subbundle of $E|_{\{z_1\} \times \bbP^1}$.
  The latter is by assumption trivial, hence in particular semistable, so $\deg( L|_{\{z_1\} \times \bbP^1}) \leq 0$ follows; this means $b \leq 0$.

  Restricting also to $\bbP^1 \times \{z_2\}$ for a general point $z_2 \in \bbP^1$, the analogous argument shows $a \leq 0$. Moreover, $L$ is nontrivial because $E$ has non nonzero global section, so $(a, b) \neq (0, 0)$.
  Hence $H_B \cdot c_1( L) = sb + ta < 0$ for every ample divisor class $H_B = (s, t)$ with $s, t > 0$ on $B = \bbP^1 \times \bbP^1$.
\end{proof}
\begin{corollary}
  Let integers $k_0, k_{\infty} \geq 1$ and epimorphisms
  \beqa
    p_0: \calO_{\bbP^1}^2 \twoheadrightarrow \calO_{\bbP^1}( k_0) \qquad\text{and}\qquad p_{\infty}: \calO_{\bbP^1}^2 \twoheadrightarrow \calO_{\bbP^1}( k_{\infty})
  \eeqa
  be given, with $\ker( p_0) \neq \ker( p_{\infty})$. Let $E$ be the kernel of the composition
  \begin{equation} \label{eq:hecke}
    \calO_{\bbP^1 \times \bbP^1}( 1, 0)^2 \xrightarrow{(q \times \id)^2} \calO_{\{0, \infty\} \times \bbP^1}^2
      \xrightarrow{p_0 \oplus p_{\infty}} \calO_{\{0\} \times \bbP^1}( k_0) \oplus \calO_{\{\infty\} \times \bbP^1}( k_{\infty})   
  \end{equation}
  where the first map is induced by an epimorphism $q: \calO_{\bbP^1}( 1) \twoheadrightarrow \calO_{\{0, \infty\}}$ of coherent sheaves on $\bbP^1$.
  Then $E$ is a vector bundle of rank $2$ over $B = \bbP^1 \times \bbP^1$ with $c_1( E) = 0$ and $c_2( E) = k_0 + k_{\infty}$, which is slope $H_B$-stable for every ample divisor class $H_B$ on $B$.
\end{corollary}
\begin{proof}
  The sheaves in \eqref{eq:hecke} are coherent over $B = \bbP^1 \times \bbP^1$ and flat over the second factor.
  Hence $E$ is also coherent over $B$, flat over the second factor, and its restriction to $\bbP^1 \times \{z_2\}$ is a subsheaf of $\calO_{\bbP^1 \times \{z_2\}}( 1)^2$ for every point $z_2 \in \bbP^1$.
  In particular, these restrictions of $E$ are torsionfree, and hence flat, over $\bbP^1$. Now the local criterion for flatness implies that $E$ is flat over $B$, and hence a vector bundle of rank $2$.

  As $E$ is the kernel of the epimorphism \eqref{eq:hecke}, it has Chern classes $c_1( E) = 0$ and $c_2( E) = k_0 + k_{\infty}$.
  To prove stability, we check that $E$ satisfies the hypotheses of the previous lemma.

  The epimorphism \eqref{eq:hecke} induces a monomorphism on global sections, since $\rmH^0( q \times \id)$ is an isomorphism,
  and $\rmH^0( p_0)$, $\rmH^0( p_{\infty})$ are both injective due to the assumption $k_0, k_{\infty} \geq 1$. Consequently, the kernel $E$ has no global sections.

  Given any point $z_1 \in \bbP^1 \setminus \{0, \infty\}$, the restriction of $E$ to $\{z_1\} \times \bbP^1$ coincides with the restriction of $\calO_{\bbP^1 \times \bbP^1}( 1, 0)^2$ to $\{z_1\} \times \bbP^1$;
  the latter is a trivial vector bundle.

  Now let a point $z_2 \in \bbP^1$ be given. Trivialising the line bundles $\calO_{\bbP^1}( k_0)$ and $\calO_{\bbP^1}( k_{\infty})$ in $z_2$,
  the restriction of $E$ to $\bbP^1 \times \{z_2\}$ becomes the kernel of the composition
  \begin{equation} \label{eq:hecke_y}
    \calO_{\bbP^1}( 1)^2 \xrightarrow{q \oplus q} \calO_{\{0, \infty\}}^2 \xrightarrow{p_{0, z_2} \oplus p_{\infty, z_2}} \calO_{\{0\}} \oplus \calO_{\{\infty\}}.
  \end{equation}
  By assumption, $\ker( p_0) \neq \ker( p_{\infty})$ as line subbundles in $\calO_{\bbP^1}^2$. Hence $\ker( p_{0, z_2}) \neq \ker( p_{\infty, z_2})$ as lines in $\bbC^2$ if the point $z_2 \in \bbP^1$ is general, 
  and thus $\bbC^2 = \ker( p_{0, z_2}) \oplus \ker( p_{\infty, z_2})$; consequently, the epimorphism \eqref{eq:hecke_y} decomposes into the direct sum of two epimorphisms
  \begin{equation}
    \calO_{\bbP^1}( 1) \longto \calO_{\{0\}} \qquad\text{and}\qquad \calO_{\bbP^1}( 1) \longto \calO_{\{\infty\}}.
  \end{equation}
  The kernels of these epimorphisms are isomorphic to $\calO_{\bbP^1}$. This shows that the kernel of \eqref{eq:hecke_y},
  and hence the restriction of $E$ to $\bbP^1 \times \{z_2\}$, is isomorphic to $\calO_{\bbP^1}^2$ if $z_2$ is general.
\end{proof}
Now let $\tau$ denote the involution $z \mapsto -z$ of $\bbP^1 = \bbC \cup \{ \infty\}$, and let $\tau_B = \tau \times \tau$ denote the induced involution of $B = \bbP^1 \times \bbP^1$.
\begin{corollary}
  For every integer $k \geq 2$, there exists a $\tau_B$-invariant vector bundle $E$ of rank $2$ over $B = \bbP^1 \times \bbP^1$ with $c_1( E) = 0$ and $c_2( E) = k$
  such that $E$ is slope $H_B$-stable for every ample divisor class $H_B$ on $B$.
\end{corollary}
\begin{proof}
  Choose a decomposition $k = k_0 + k_{\infty}$ with $k_0, k_{\infty} \geq 1$.

  We lift $\tau$ to an involution $\tau_1$ on the line bundle $\calO_{\bbP^1}( 1)$ over $\bbP^1$ such that
  $\tau_1 = \id$ in the fiber over the fixed point $0 \in \bbP^1$, and $\tau_1 = -\id$ in the fiber over $\infty \in \bbP^1$.
  In terms of the homogenous coordinates $[w:z]$ on $\bbP^1$ with $[0:1] = \infty$ and $[1:0] = 0$, the corresponding global sections $w$ and $z$ of $\calO_{\bbP^1}( 1)$ satisfy $\tau_1( w) = w$ and $\tau_1( z) = -z$.
  This involution $\tau_1$ induces an involution $\tau_n := \tau_1^{\otimes n}$ on the line bundle $\calO_{\bbP^1}( n)$ over $\bbP^1$ for all $n \in \bbZ$.

  If $n \geq 1$, then $\calO_{\bbP^1}( n)$ can be globally generated by two eigensections $s_n^+, s_n^- \in \rmH^0( \calO_{\bbP^1}(n))$ with $\tau_n( s_n^+) = s_n^+$ and $\tau_n( s_n^-) = -s_n^-$.
  (For example, one can take
  \begin{align*}
    s_{2n-1}^+ & := w^{2n-1}, & s_{2n}^+ & := w^{2n} + z^{2n}\\
    s_{2n-1}^- & := z^{2n-1}, & s_{2n}^- & := w^{2n-1} \cdot z
  \end{align*}
  for $n \geq 1$.) Such sections define epimorphisms
  \beqa
    p_0 := s_{k_0}^+ \oplus s_{k_0}^-: \calO_{\bbP^1}^2 \twoheadrightarrow \calO_{\bbP^1}( k_0) \qquad\text{and}\qquad
    p_{\infty} := s_{k_{\infty}}^- \oplus s_{k_{\infty}}^+: \calO_{\bbP^1}^2 \twoheadrightarrow \calO_{\bbP^1}( k_{\infty})
  \eeqa
  which have nonisomorphic images for $k_0 \neq k_{\infty}$, and are not proportional for $k_0 = k_{\infty}$; hence $\ker( p_0) \neq \ker( p_{\infty})$ in any case.

  The involution $\tau_1$ on $\calO_{\bbP^1}( 1)$ also induces by pullback an involution $\tau_{1, 0}$ on the line bundle $\calO_{\bbP^1 \times \bbP^1}( 1, 0)$ over $B$.
  With our choices for $p_0$ and $p_{\infty}$, the composition \eqref{eq:hecke} is equivariant with respect to the involutions $-\tau_{1, 0} \oplus \tau_{1, 0}$ and $\tau_{k_0} \oplus \tau_{k_{\infty}}$.
  Hence its kernel is $\tau_B$-invariant; due to the previous corollary, this kernel is a stable vector bundle $E$ of rank $2$ over $B$ with the required Chern classes.
\end{proof}

\section{Stable $\tau_X$-invariant rank four bundles}

Let $E_1$ and $E_2$ be two $H_B$-stable $\tau_B$-invariant rank two bundles 
on $B$ with $c_1(E_i)=0$ and $c_2(E_i)=k_i\geq 2$. We first note 
\begin{lemma}
$\pi^*E_i$ is stable on $X$ with respect to $H=z\Sigma+\pi^*H_B$ with $H_B=hc_1$
and $0<z<h$.
\end{lemma}
The proof is given in (\cite{AC1}, Section 4.1).

Let $D$ be an invariant divisor on $X$. We consider the extension defining the rank 4 bundle with $D=x\Sigma+\pi^*\alpha$
\beqa\label{ext}
0\to \pi^*E_1(D)\to V_4\to\pi^*E_2(-D)\to 0.
\eeqa
Following the above construction we find that $V_4$ is stable with respect 
to $H_{\epsilon}=H+\epsilon D$ if the conditions \ref{A}, \ref{B}, \ref{C} are satisfied.

To solve \ref{A} we first note that we have either to choose $x<0$ and $\alpha c_1>0$ or 
$x>0$ and $\alpha c_1<0$ (or $x=0$ and $\alpha c_1=0$). We find 
\beqa\label{condA}
\ref{A}\colon \ \ \frac{h^2}{(h-z)^2}=1-\frac{xc_1^2}{\alpha c_1}
\eeqa
Thus given $x, \alpha$ and $c_1$ we can always solve for $h,z$.  

The next step is to compute $\chi_D(E_2,E_1)$ and solve $\chi_D(E_2,E_1)<0$.
If we insert the expressions for $D$ and the Chern classes of $\pi^*E_i$ into (\ref{eq:HRR})
we find the condition
\beqa
\ref{B}\colon \ \ (8x^3-x)c_1^2+6(1-4x^2)\alpha c_1+x(24\alpha^2+c_2)-6x(k_1+k_2)<0.
\eeqa
Finally, we have to assure that condition \ref{C} is satisfied
\beqa 
\ref{C}\colon\ \ \Sigma\cdot(x(h-z)c_1+z\alpha)+h \alpha\cdot c_1\not\equiv 0.
\eeqa
Since we assume $\alpha c_1\neq 0$ this condition is solved.

We also note that the Chern classes of $V_4$ are given by 
\begin{align}
c_1(V_4)&=0,\\
c_2(V_4)&=-2x\pi^*(2\alpha-xc_1)\cdot \Sigma+(k_1+k_2-2\alpha^2)[F],\\
c_3(V_4)&=4x(k_2-k_1).
\end{align}
where $[F]$ denotes the curve class of the fiber of $X$.

Now any simultaneous solution to \ref{A}, \ref{B} and \ref{C} gives a $\tau_X$-invariant 
$H_\epsilon$-stable rank four bundle which will be the starting point for the rank five 
bundle construction in the next section.

\section{Stable $\tau_X$-invariant rank five bundles}

We will now construct an invariant stable rank 5 bundle applying again the above construction.

Let $\D$ be another invariant divisor in $X$. We consider the extension defining the rank five bundle
with $\D=y\Sigma+\pi^*\beta$ and the given rank four bundle $V_4$ 
\beqa\label{ext2}
0\to V_4(\D)\to V_5 \to {\mathcal O}_X(-4\D)\to 0.
\eeqa
As above $V_5$ will be stable with respect to $H_{\epsilon, \bar\epsilon}=H_\epsilon+\bar\epsilon \D$
if the conditions \ref{A}, \ref{B} and \ref{C} are satisfied. We find 
\begin{align}
\ref{A}\colon\ \ 0=yh^2c_1^2+z(2h-z)(\beta-yc_1)c_1&+2\epsilon\Big[z\alpha\beta+(h-z)\big(x(\beta-yc_1)c_1+y\alpha c_1\big)\Big]\nonumber\\
&+ \epsilon^2\Big[2x\alpha\beta+y\alpha^2-x\big(x(\beta-yc_1)c_1+2xy\alpha c_1\big)\Big]\nonumber
\end{align}
Now $z,h$ have been fixed in the construction of $V_4$. Also we can not solve for $\epsilon$ as this
is chosen sufficiently small thus to solve this constraint we have to make an assumption on
the intersection numbers between $\alpha$, $\beta$ and $c_1$ and also $y$. We find condition 
\ref{A} is solved if we take
\beqa \label{ansatz}
y=0, \ \ \alpha \beta=0, \ \ \beta c_1=0.
\eeqa
Below we will see that the choice $y=0$ will actually be enforced by the solvability of the anomaly 
constraint. With (\ref{ansatz}) the nonsplit condition \ref{B} simplifies and we get the condition
\beqa
\ref{B}\colon\ \ x(k_2-k_1)<0
\eeqa
The last condition we need to solve is 
\beqa
\ref{C}\colon\ \ \Sigma \beta (z+\epsilon x)\not\equiv 0.
\eeqa
As we assume $\beta\neq 0$ the condition is solved.

The Chern classes of the rank five bundle $V_5$ are given by
\begin{align}
c_1(V_5)&=0,\\
c_2(V_5)&=c_2(V_4)-10\D^2,\\
c_3(V_5)&=c_3(V_4)-20\D^3-2c_2(V_4)\cdot \D.
\end{align}

\section{Solutions with $[W]\neq 0$}

In this section we assume that a stable rank 5 bundle $V_5$ in the visible sector is specified 
and take the trivial bundle as hidden sector bundle and so leave the hidden $E_8$ gauge group 
unbroken. The second Chern class constraint will be solved in this section by allowing 
a number of background fivebranes in the heterotic vacuum corresponding to the wrapping
of an effective curve class $[W]$.

Given a stable rank 5 bundle $V_5$ we have to satisfy the heterotic anomaly equation
$c_2(X)-c_2(V_5)=[W]=\pi^*(w_B)\Sigma+a_f [F]$ where $W$ is a space-time filling fivebrane
wrapping a holomorphic curve of $X$. This leads to the condition that $[W]$ is an effective curve class
in $X$ which can be expressed by the two conditions: $w_B$ an effective curve class in $B$ or zero and $a_f\geq 0$. If we insert the expressions for $c_2(X)$ and $c_2(V_5)$ we get
\begin{align}
w_B&=(6-2x^2-10y^2)c_1+4x\alpha+10y\beta\geq 0,\\ 
a_f&=44-k_1-k_2+2\alpha^2+10\beta^2\geq 0.
\end{align}
From the construction of $V_4$ and $V_5$ we find that the curve classes $x\alpha$ and $y\beta$
are both negative effective or have at least one negative entry. We thus conclude that the only possibility 
to solve $w_B\geq 0$ is to take $x=\pm 1$ and $y= 0$. Moreover, for $x=\pm 1$ we get $w_B\geq 0$
for $c_1\pm \alpha\geq 0$ which constrains possible $\alpha$ classes. 

We write divisor classes on $B$ in the form $(p,q)$ where the entries refer to the two generators in $B={\mathbb P}^1\times{\mathbb P}^1$.
From $w_B\geq 0$ and the condition (\ref{condA}) we find an $\alpha$ class has to satisfy for 
\begin{align}\label{aclass}
x=1&\colon \ \ c_1+\alpha\geq 0, \ \  \alpha\cdot c_1<0,\\
x=-1&\colon \ \ c_1-\alpha\geq 0, \ \ \alpha\cdot c_1>0,\nonumber
\end{align}
which leads to a list of possible $\alpha$ classes. However, before working out this list we 
note that the possible $\beta$ classes are constraint by the condition $a_f\geq 0$ and $\beta c_1=0$. 
This leads to the two possible classes
\beqa\label{bclass}
\beta=(-1,1)\ \ {\rm and}\ \ (1,-1).
\eeqa
Now in summary an $\alpha$ class has to satisfy (\ref{aclass}) together with $\alpha\beta=0$ where $\beta$ is given by (\ref{bclass}). This leads to the list of possible $\alpha$ classes depending on $x$
\begin{align}\label{aclasssol}
x=1&\colon\ \  \alpha=(-2,-2), (-1,-1),\\
x=-1&\colon\ \ \alpha=(2,2), (1,1).\nonumber
\end{align}
In summary, we are left to solve the following system of conditions with $\alpha$, $\beta$ and
$x$ as given in (\ref{bclass}) and (\ref{aclasssol})
\begin{align}
(8x^3-x)c_1^2+6(1-4x^2)\alpha c_1+x(24\alpha^2+c_2)-6x(k_1+k_2)&<0,\\
x(k_2-k_1)&<0,\\
a_f=24-k_1-k_2+2\alpha^2&\geq 0,\\
c_3(V_5)/2=2x(k_2-k_1)&=\pm 6,
\end{align}
where the first two inequalities are the non-split conditions for the defining extensions of 
$V_4$ and $V_5$, respectively. We also recall that $k_i\geq 2$.

We find the following solutions:
\begin{align}
x=-1&\colon \ \  \alpha=(2,2),\ \ k_1=2+i,\ \ k_2=5+i\ \ (i=0,\dots,16),\\
 &\ \ \ \ \alpha=(-1,-1), \ \ k_1=2+i,\ \ k_2=5+i\ \ (i=0,\dots,10).
\end{align}

\section{Solutions with $[W]=0$}

The solutions with $\alpha=c_1$ led to $w_B=0$ but with $a_f>0$. 
In this section we will now specify a hidden sector bundle $V_{hid}$ with second Chern class $c_2(V_{hid})=k_3$ and solve the anomaly equation without fivebranes, i.e., $[W]=0$. More precisely,
we want to specify a polystable invariant bundle $V_5\oplus V_{hid}$ of vanishing first Chern class and whose second Chern class satisfies 
\beqa\label{polyano}
 c_2(X)=c_2(V_5)+c_2(V_{hid}).
\eeqa 
We will take as hidden sector bundle $V_{hid}=\pi^*E_{hid}$ with $\rk(E_{hid})=2$ using again the construction given in Section 4. This bundle is $H$-stable and $\tau_X$-invariant 
and has $c_2(\pi^*E_{hid})=k_3$ with $k_3\geq 2$. Since $H$-stability is an open condition in $H$ according to \cite[Appendix 4.C]{HL}, $\pi^*E_3$ is still stable with respect to $H_{\epsilon}=H+\epsilon D$
if $\epsilon$ is sufficiently small, and then also with respect to $H_{\epsilon, \bar\epsilon} = H_{\epsilon} + \bar\epsilon\D$ if $\bar\epsilon$ is also sufficiently small (possibly depending on $\epsilon$).
Thus the bundle $V_5\oplus V_{hid}$ is polystable as required and we obtain solutions to the above system with $[W]=0$ for $x=-1$ and $\alpha=(2,2)$, $\beta=(-1,1)$ and
\beqa
k_1=2+i, \ \  k_2=5+i, \ \ k_3=33-2i \ \ \  (i=0,\dots,15).
\eeqa
Moreover, as $V_{hid}$ is a pullback bundle of rank two we get $c_3(V_{hid})=0$ and no chiral matter in the hidden sector. 

This class of solutions can now be perturbed  
to a solution of the Strominger system using the results of \cite{AGF1} and \cite{AGF2} and should provide 
a step toward phenomenologically interesting heterotic flux compactifications via non-K\"ahler deformation of Calabi-Yau geometries via polystable bundles.

\end{document}